\documentclass[12pt]{article}
\usepackage[english]{babel}
\usepackage{amsmath,amsthm,eepic}
\usepackage{amsfonts}
\usepackage{array}
\newtheorem{lemma}{Lemma }%[section]

\newcommand{\bigsize}{\fontsize{16pt}{20pt}\selectfont}
\def\openone{\leavevmode\hbox{\small1\kern-3.3pt\normalsize1}}
\def\im{\mbox{Im\,}}

\def\ad{\mbox{ad\,}}

\begin{document}

\begin{center}
{\bf\bigsize
\scshape
Soliton equations related to the affine Kac-Moody algebra $D^{(1)}_{4}$}\\
\vspace{8mm}
V. S. Gerdjikov$^1$, D.M. Mladenov $^{2}$, A.A. Stefanov $^{2,3}$, S.K. Varbev $^{2,4}$\\[6mm]
\small{
$^{1}$ Institute of Nuclear Research and Nuclear Energy,
Bulgarian Academy of Sciences,
72 Tsarigradsko chausse, Sofia 1784, Bulgaria. \\[3mm]
$^{2}$ Theoretical Physics Department, Faculty of Physics,
Sofia University "St. Kliment Ohridski",
5 James Bourchier Blvd, 1164 Sofia, Bulgaria. \\[3mm]

$^{3}$ Institute of Mathematics and Informatics,
Bulgarian Academy of Sciences,
Acad. Georgi Bonchev Str., Block 8, 1113 Sofia, Bulgaria.

$^{4}$ Institute of Solid State Physics,
Bulgarian Academy of Sciences, 72 Tzarigradsko chaussee, Sofia 1784, Bulgaria. \\[3mm]
}
\vspace{8mm}
\end{center}

\begin{abstract}
We have derived the hierarchy of soliton equations associated with the untwisted affine Kac-Moody algebra $D^{(1)}_{4}$ by calculating the corresponding recursion operators.
The Hamiltonian formulation of the equations from the hierarchy is also considered.
As an example we have explicitly presented the first non-trivial member of the hierarchy, which is an one-parameter family of mKdV equations.
We have also considered the spectral properties of the Lax operator and introduced a minimal set of scattering data.
\end{abstract}

%-----------------------------------------------------------------
\section{Introduction}
\label{intro}
After the discovery of the inverse scattering method (ISM) in 1967  by the Princeton group of  Gardner, Greene, Kruskal and
Miura \cite{GGKM} the study of soliton equations and the related integrable nonlinear evolution equations (NLEEs)  became one of the most active branches of nonlinear science.
In this pioneering paper, the Korteveg-de Vries (KdV) equation was exactly integrated using the ISM.
One year later, Peter Lax presented a method, which would later be instrumental in the search of new integrable equations different from the KdV\cite{Lax}.
This started numerous attempts to extend the range of application of the ISM to other equations.

Three years had passed, when Gardner started the canonical approach to the NLEEs.
He was the first to realize that the KdV equation can be written in a Hamiltonian form \cite{Gardner}.
Soon after that, Zakharov and Faddeev proved that KdV equation is a completely integrable Hamiltonian system \cite{Zakharov-Faddeev}
and Wadati solved the modified Korteweg-de Vries (mKdV) equation \cite{Wadati-MKdV}.

This explosion of interest in the soliton science led to the discovery of a completely new world of unknown before integrable equations
together with methods and techniques for finding their exact solutions.
We should mention some of those methods and techniques, starting with the fundamental analytic solution (FAS), which was introduced by Shabat \cite{Sha}.
Within this approach one can show that the ISM is equivalent to a multiplicative Reimman-Hilbert problem (RHP),
which is the basis for the dressing method invented by Zakharov and Shabat \cite{ZakharovShabat1,ZakharovShabat2}.
Another important result was derived by Ablowitz, Kaup, Newell and Segur \cite{AKNS}.
They showed that the ISM can be  interpreted as a generalized Fourier transform.
This approach is based on the Wronskian relations, which are basic tools for the analysis of the mapping between the potential  and the scattering data.
They allow the construction of the "squared solutions" of the Lax operator, which form a complete set of functions and play the role of the generalized exponentials (see \cite{IP2,GeYa*94,GeYaV} and the references therein).
Another powerful technique we should mention here is that of the so called recursion operators.
This method has in principal a long history but the application to the NLEEs started in
\cite{Gardner-Greene-Kruskal-Miura} (see also \cite{McKean-vanMoerbeke} and \cite{Its-Matveev})
where with the help of the recursion operator the KdV equation was written in a compact form.
Within the study of the Nonlinear Schr\"{o}dinger (NS) equation, considered in the rapidly decreasing case and imposing the requirement
that the squared Jost solutions of the auxiliary linear problem are its eigenfunctions, the recursion operator was first introduced in \cite{AKNS}.
The case of the general matrix linear first order differential operator was studied in \cite{Newell-RNM}.
We should also mention the Tzitzeica equation \cite{Tzitzeica,VKT}, which was rediscovered as a soliton equation at the end of the 1970s.
Its Lax pair was discovered by Mikhailov \cite{Mikhailov} who demonstrated, that Tzitzeica equation has a hidden $\mathbb{Z}_3$-symmetry, which becomes evident at the Lax pair level.
This led Mikhailov to the discovery of the reduction group, an important concept in the Lax representation of soliton equations.
He also solved the family of 2-dimensional Toda field theories (of which Tzitzeica equation is a particular example) by using the dressing method.
The further development of all these methods and techniques and also their numerous applications can be found for example in
\cite{Ablowitz-Clarkson,Ablowitz-Segur,Bullough-Caudrey,Calogero-Degasperis,Dodd-Eilbeck-Gibbon-Morris,Faddeev-Takhtadjan,Lamb,Newell-book,NMPZ,X1}
 and more recently in \cite{X2,VG-Ya-13,VG-Ya-14}.
For a brief introduction to the very intriguing history of the solitons and the ISM see, for instance,
\cite{Allen-SolitonHistory,Miles-SolitonHistory}
and also
\cite{Bullough-Caudrey,Dodd-Eilbeck-Gibbon-Morris,Newell-book}.

The theory of soliton equations deals with a class of nonlinear partial differential equations that admit exact solutions.
That is why it is natural to study the soliton equations from a group-theoretical point of view and to uncover the algebraic structures,
which are behind them.
Investigations clarifying the algebraic structure of several classes of NLEEs were, for the first time, started in the early 1980s in the Research Institute for Mathematical Sciences by the Kyoto group.
Namely, M. Sato and Y. Sato in their studies discovered the very remarkable fact that there is a transitive action of
an infinite-dimensional group on the manifold of solutions of the Kadomtsev-Petviashvili (KP) hierarchy and the corresponding
Lie algebra is $\mathfrak{gl}(\infty)$ \cite{Sato,Sato-Sato}.

Another very important investigation, which shows deep and fundamental links between NLEEs and infinite-dimensional
Lie groups and Lie algebras was done by Drinfeld and Sokolov \cite{Drinfeld-Sokolov-1981,Drinfeld-Sokolov-1985}.
In their papers, starting from an affine Kac-Moody algebra and applying the method of Hamiltonian reduction, they obtained
a one-parameter family of Poisson brackets, which are defined on some infinite-dimensional Poisson manifold,
and constructed an integrable hierarchy of bi-Hamiltonian equations.
In their studies they found most of the mKdV equations related to the low dimensional cases,
for example, when the affine Kac-Moody algebra is $A^{(1)}_{1} \simeq \hat{\mathfrak{sl}}_2$
the resulting equation is the well-known mKdV equation. Within this scheme, Casati, Della Vedova and Ortenzi
\cite{Casati-DellaVedova-Ortenzi}
found in an explicit form the hierarchy, which is related to the affine Kac-Moody algebra $G_2^{(1)}$, together with the bi-Hamiltonian structure
of the obtained soliton equations.

These investigations were followed by Date, Jimbo, Kashiwara and Miwa with a sequance of seminal papers
\cite{Date-Jimbo-Kashiwara-Miwa} (see also \cite{Jimbo-Miwa-Date-Solitons} and references there in) in which
the representation theory of affine Kac-Moody algebras is used to obtain KP and KdV hierarchies, for example,
the KdV and KP hierarchies are related correspondingly to the affine Lie algebras $A^{(1)}_{1}$ and $A_{\infty}$.
A very important side of these investigations is that they put on very general grounds the use
of the theory of affine Lie algebras and their corresponding Lie groups in finding the hidden symmetries of soliton equations.

The relation between infinite-dimensional Lie algebras, and their groups and soliton equations has inspired a huge
number of further investigations in which this link has been widely extended and generalized in many different senses and directions.
In particular, Segal and Wilson elaborated a construction, which associates with every solution of the KdV equation a point of
some infinite-dimensional Grassmannian \cite{Segal-Wilson}.
They show that the geometry of the Grassmannian reflects to the properties of the exact solutions.
Within this powerful approach the explicit algebro-geometric solutions are naturally described and the geometrical meaning of the $\tau$-function is clarified.
Kac and Wakimoto \cite{Kac-Wakimoto} constructed a scheme in which vertex operator representations of some affine Lie algebras are used to obtain hierarchies of NLEEs.
In the papers \cite{deGroot-Hollowood-Miramontes,Burroughs-deGroot-Hollowood-Miramontes}
Burroughs, de Groot, Hollowood and Miramontes used a general approach, with the help of which they constructed
integrable hierarchies of NNLEs, which in turn are related  to untwisted Kac-Moody algebras.
The generalized KdV hierarchies of Drinfeld and Sokolov can be naturally included within this approach and are obtained as some special cases of the general scheme.
They also include most of the generalizations of the original Drinfeld-Sokolov construction available in the literature.
The extent of this class of hierarchies by means of (matrix) pseudo-differential operators has been investigated by
Feh\'{e}r, Harnad, Marshall, Delduc and Gallot
\cite{Delduc-Feher,Delduc-Feher-Gallot,Fehеr-Harnad-Marshall,Feher-Marshall}
(see also \cite{Aratyn-Ferreira-Gomes-Zimerman}).
Their construction and its generalizations provide a systematic approach to the study and classification of many integrable hierarchies
described by means of pseudo-differential Lax operators.
For a brief but very informative review of some of these results see for example \cite{Fehеr}.
An equivalent construction to the generalized Drinfeld-Sokolov hierarchies was proposed by Feigin and Frenkel \cite{Feigin-Frenkel-1,Feigin-Frenkel-2}.
In this approach, instead of the hamiltonian flows of the mKdV hierarchy, the main technical tool is given by some properly chosen vector fields
defined on the space of jets of the potential of the Lax operator. This potential in turn belongs to the Cartan subalgebra of the corresponding
affine Kac-Moody algebra.
Enriquez and Frenkel \cite{Enriquez-Frenkel} proved that this construction is
equivalent to that proposed previously by Drinfeld and Sokolov.

The present paper can be considered as a natural continuation on the investigations of the relation between Kac-Moody algebras and integrable equations started in our previous works \cite{GMSV1,GMSV2,GMSV3,GMSV4}.
In particular, there we have analyzed the family of integrable nonlinear third order partial differential equations associated with the untwisted affine Kac-Moody algebras $A^{(1)}_{r}$, which represent the third member of the corresponding hierarchy.
% Further in that papers we presented some particular examples, imposed additional reductions and also
We also started investigations concerning the affine Kac-Moody algebra $D_4^{(1)}$.

The $D_4\simeq so(8)$ is unique among the simple Lie algebras, because  $D_4$ is the only algebra with an $S_3$ symmetry of its Dynkin diagram.
It is also the only simple Lie algebra that has 3 as a double-valued exponent.
For these reasons, the main purpose of this paper is to derive directly% in an explicit way
the hierarchy of soliton equations related to the untwisted affine Kac-Moody algebra $D_4^{(1)}$ by calculating the corresponding recursion operators. This is an alternative approach to the one used in
 \cite{Drinfeld-Sokolov-1981,Drinfeld-Sokolov-1985}.
% to write down the integrable systems in Hamiltonian form and give the general formula for their Hamiltonians.
In particular, we will pay special attention to the one-parameter family of mKdV equations.
We also formulate the spectral properties of the Lax operators as well as the FAS of the system.

The paper is organized as follows.
In Section 2 we give some basic facts concerning Kac-Moody algebras and present an explicit way of constructing a basis in the algebra $D_4^{(1)}$.
In Section 3 the Lax pair corresponding to the hierarchy is defined and then analyzed and the recursion operators giving the hierarchy of soliton equations are derived.
In Section 4, following \cite{Drinfeld-Sokolov-1981,Drinfeld-Sokolov-1985},
 we give the Hamiltonian formulation of the obtained soliton equations along with
the general formula for the corresponding Hamiltonians.
In Section 5 we derive  explicitly the first non-trivial member of the hierarchy, which is a one-parameter family set of mKdV type equations.
In Section 6 we formulate the spectral properties of the Lax operators and the corresponding RHP. We construct two minimal sets of scattering data
and prove that they allow one to reconstruct both the scattering matrix and the potential of the Lax operator.
Section 7 contains concluding remarks.
In the Appendix we present some properties of the simple Lie algebra $D_4$ and show an explicit construction of the dihedral realization of the Coxeter automorphism.
Here, we also give the explicit form of the basis of $D_4^{(1)}$, which is used in the calculations. We also show the explicit calculation of $\mbox{ ad}_J^{-1}$, a step
that is required for the calculation of the recursion operators.

\section{Preliminaries}
\label{prelim}
%Lie algebras
A Lie algebra
 %======================================================================================================%
\footnote{
For a comprehensive treatment concerning finite-dimensional Lie algebras,
as well as infinite-dimen\-sio\-nal Lie algebras of Kac-Moody type, see %, for instance,
\cite{Carter,Goto-Grosshans,Helgasson,Jacobson,Kac}.
A collection of articles can be found in
\cite{Sthanumoorty-Misra}.
}
%======================================================================================================%
 is called simple if it is non-Abelian and contains no nontrivial ideals.
 A semisimple Lie algebra is the direct sum of finite number of simple Lie algebras.
 We will consider algebras over the field of complex numbers $\mathbb{C}$.

%Ad_X
Let $\mathfrak{g}$ be a finite-dimensional simple Lie algebra.
By $\mbox{ad}_{X}$ we denote the linear operator defined by
\begin{equation}
\mbox{ad}_{X}(Y) = \big[X,\, Y \big], \quad X,Y \in \mathfrak{g}.
\end{equation}
This operator has a kernel and can only be inverted on its image.
We denote that inverse by $\mbox{ad}_X^{-1}$.
If $X$ is diagonalizable then $\mbox{ad}_X^{-1}$ can be expressed as a polynomial of $\mbox{ad}_X$.
%Cartan form
Let  $\big\langle\; ,\; \big\rangle$ be the Killing-Cartan form on $\mathfrak{g}$, which is  defined by
\begin{equation}
\big\langle X ,Y \big\rangle = \mbox{tr} \left( \mbox{ad}_{X} \mbox{ ad}_{Y} \right).
\end{equation}
Since any invariant symmetric bilinear form on $\mathfrak{g}$ is proportional to the Killing form we will use the form given by
\begin{equation}
\big\langle X ,Y \big\rangle = \mbox{tr} \left( {X}{Y} \right).
\end{equation}

% Automorphisms
Let $\varphi$ be an automorphism of $\mathfrak{g}$ of finite order.
If it is of the form
\begin{equation}
\varphi(X)= e^{F}Xe^{-F}
\end{equation}
for some generator $F$ then it is called an inner automorphism.
Every automorphism that is not inner is an outer automorphism.
The set of outer automorphisms of $\mathfrak{g}$ is equivalent, up to a conjugation with an inner automorphism,
to the symmetries of the Dynkin diagram of $\mathfrak{g}$.
% Graded algebras
Every finite-order automorphism $\varphi$ introduces a grading in $\mathfrak{g}$ by
\begin{equation}
\mathfrak{g}=\mathop{\oplus}\limits_{k=0}^{s-1} \mathfrak{g}^{(k)},
\end{equation}
such that
\begin{equation}
\label{grading}
\big[ \mathfrak{g}^{(k)}, \mathfrak{g}^{(l)} \big] \subset \mathfrak{g}^{(k+l)}, \quad \varphi(X) = \omega^k X, \quad  \omega = \exp \left(\frac{2 \pi i}{s} \right),\quad  \forall X \in \mathfrak{g}^{(k)},
\end{equation}
where $s$ is the order of $\varphi$ and $k+l$ is taken modulo $s$.
% Kac-Moody algebras
If $\mathfrak{g}$ is a finite-dimensional Lie algebra over $\mathbb{C}$ then
\begin{equation}
\begin{aligned}
\mathfrak{g}[\lambda, \lambda^{-1}]&=\left \{ \sum_{i=n}^{m} v_i \lambda^i : v_i \in \mathfrak{g} , n,m \in \mathbb{Z} \right \}, \\
f(\lambda) &= \left \{ \sum_{i=0}^{m} f_i \lambda^i : f_i \in \mathfrak{g} , m \in \mathbb{Z} \right \}.
\end{aligned}
\end{equation}
There is a natural Lie algebraic structure on $\mathfrak{g}[\lambda, \lambda^{-1}]$.
Let $\varphi$ be an automorphism of $\mathfrak{g}$ of order $s$. Then
\begin{equation}
L(\mathfrak{g}, \varphi) =  \left \{ f \in \mathfrak{g}[\lambda, \lambda^{-1}]: \varphi(f(\lambda))= f \left( \lambda \exp \left({\frac{2\pi i}{s}}\right) \right) \right \}.
\end{equation}
$L(\mathfrak{g}, \varphi)$ is a Lie subalgebra of $\mathfrak{g}[\lambda, \lambda^{-1}]$.
If $\mathfrak{g}$ is simple then $ L(\mathfrak{g}, \varphi)$ is called a Kac-Moody algebra.
It is obvious that Kac-Moody algebras are graded algebras. Note that commonly the central extension of $L(\mathfrak{g}, \varphi )$ is called a Kac-Moody algebra.
The definition given above is the one used in \cite{Drinfeld-Sokolov-1981,Drinfeld-Sokolov-1985}.
% Dynkin diagram automorphisms
Every automorphism $\varphi$ can be uniquely represented in the form $\varphi = f \circ \varphi_{\tau}$ where $\varphi_{\tau}$ is given by
\begin{equation}
\varphi_{\tau} (E_{\alpha_i}) = E_{\tau(\alpha_i)}
\end{equation}
and $\tau$ is a permutation of the simple roots that preserves the symmetry of the Dynkin diagram of $\mathfrak{g}$, i.e. it is an automorphism of the Dynkin diagram. The order of $\varphi_{\tau}$ is called the height of $L(\mathfrak{g}, \varphi)$.
As is shown in \cite{Kac}, two Kac-Moody algebras $L(\mathfrak{g_1}, \varphi_1), L(\mathfrak{g_1}, \varphi_1)$ are equivalent if $\mathfrak{g_1}$ is isomorphic to $\mathfrak{g_2}$, and the automorphisms of the Dynkin diagram determined by $\varphi_1$ and $\varphi_2$ are conjugate.

% Coxeter automorphism
An automorphism $C$ of a simple Lie algebra is said to be a Coxeter automorphism if the eigen-space that corresponds to
an egeinvalue $1$ is Abelian and $C$ is of minimal order.
The order of $C$ is called the Coxeter number of the algebra.

%-----------------------------------------------------------------
In this paper we will consider the equations related to $ L(D_4, C)$, where $C$ is a Coxeter automorphism of $D_4$.
We will use the standard notation $D_4^{(1)}$ \cite{Carter}.
% Grading of D_4
The Coxeter number of $D_4^{(1)}$ is $6$ and its exponents are $1,3,5,3$ \cite{Drinfeld-Sokolov-1981,Carter}.
The Coxeter automorphism is realized as
\begin{equation}
C(X) = cXc^{-1},
\end{equation}
where $c$ is given in the appendix (see \eqref{coxeter1}).
It introduces a grading in $D_4$ via \eqref{grading}.
%----Basis---------------------------------------------------
A basis compatible with this grading is given by averaging the Cartan-Weyl basis over the action of the Coxeter automorphism
\begin{equation}
 \mathcal{E}_{i}^{(k)} =  \sum_{s=0}^{5} \omega^{-s k} C^{s} ( E_{\alpha_i}), \qquad  \mathcal{H}_j^{(k)} = \sum_{s=0}^{5} \omega^{-s k} C^{s} ( H_{j}). \label{basis1}
\end{equation}
Note that $\mathcal{H}_j^{(k)}$ is non-vanishing only if $k$ is an exponent.
Since $3$ is a double-valued exponent, $ \mathfrak{g}^{(3)}$ will contain two Cartan elements.
The explicit form of the basis is given in the appendix.

\section{Lax pair and recursion operators}
\label{lax pair}
%----Lax pair-------------------------------------------
%Reductions
We will derive the hierarchy of equations related to a Lax pair that is subject to a $\mathbb{Z}_h$ reduction \cite{Mikhailov}
\begin{equation} \label{LP12}\begin{aligned}
&L\psi \equiv i\partial_x \psi + U(x,t,\lambda)\psi = 0, \\
&M\psi \equiv i\partial_t \psi + V(x,t,\lambda)\psi = \psi \Gamma(\lambda),
\end{aligned}
\end{equation}
where $\Gamma(\lambda)$ is given in Subsection \ref{sdata}. The potentials must satisfy
\begin{equation}
C\left( U(x,t,\lambda) \right) = U(x,t,\omega\lambda),\quad C\left( V(x,t,\lambda) \right) = V(x,t,\omega\lambda).
\end{equation}
with $\omega =e^{\frac{2\pi i}{6}} $. This implies that the potentials are elements of the Kac-Moody algebra $D_4^{(1)}$. A Lax pair compatible with the above reduction is given by
\begin{equation}
\label{LaxPair}
\begin{aligned}
L &= i \partial_x + Q(x,t) - \lambda J, \\
M &= i \partial_t  + \sum_{k=0}^{n-1} \lambda^k V^{(k)}(x,t) - \lambda^{n}K,
\end{aligned}
\end{equation}
where
\begin{equation}
\label{potential}
Q(x,t) \in \mathfrak{g}^{(0)}, \quad V^{(k)}(x,t) \in \mathfrak{g}^{(k)}, \quad K \in \mathfrak{g}^{(n)}, \quad J\in \mathfrak{g}^{(1)}.
\end{equation}
For simplicity, we assume that $n \leq h$ and $n$ is an exponent of $D_4^{(1)}$. If $n$ is not an exponent the resulting equations are trivial and if $n > h$ then every index should be understood modulo $h$.
The explicit form of the potential $Q$ is
\begin{equation}
Q(x,t)=\sum_{i=1}^{r} q_i(x,t)\mathcal{E}_{i}^{(0)}.
\end{equation}
To simplify the notation we will omit writing any explicit dependence on $x$ or $t$.

%-------J and K--------------------------------------------
The elements $J$ and $K$ are some properly chosen constant matrices. In our case
\begin{equation}
\label{J}
J =\frac{1}{2} \mathcal{H}^{(1)}_1  = \mbox{diag}(1,\omega,\omega^5,0,0,-\omega^5,-\omega,-1)
\end{equation}
and $K$ is chosen so that $ \big[ J, K \big] = 0$.

%-----AdJ^(-1)---------------------------------------
Our choise of $J$ fixes the inverse of $\mbox{ad}_J$ to (for more details, see the Appendix)
\begin{equation}
\mbox{ad}_J^{-1}=\frac{1}{27} \left(26 \mbox{ ad}_J^{5}+ \mbox{ad}_J^{11} \right).
\end{equation}

%----Compatibility condition------------------------
The Lax pair must commute, i.e.
\begin{equation}
\label{LPC}
\big[L,\, M \big]=0
\end{equation}
for every $\lambda$.
%----------------------------------------------
This implies the following recursion relations
%-----Recursion relations----------------------
\begin{equation}
\begin{aligned}
&\lambda^{n+1}: & \big[ J, K \big]&=0, \\
&\lambda^{n}:    & \big[ J,V^{(n-1)} \big] + \big[ Q, K \big]&=0, \\
&\lambda^{s}:     &  i{\partial_x V^{(s)}}+ \big[Q ,V^{(s)} \big]- \big[ J,V^{(s-1)} \big] &=0, \\
&\lambda^{0}:    &  - i{\partial_t Q}+  i{\partial_x V^{(0)}}+  \big[ Q(x,t),V^{(0)} \big]&=0. \\
\end{aligned} \label{Recurrence1}
\end{equation}
%-----Splitting-----------------------
Each element splits into "orthogonal" and "parallel" parts
\begin{equation}
\begin{aligned}
V^{(s)} &=V^{(s)}_{\bot}  + V^{(s)}_{\|} , \quad \mbox{ad}_J \left(V^{(s)}_{\|} \right)=0, \\
V^{(s)}_{\|}(x,t)&=
 \begin{cases}
0 &\mbox{if } s \mbox{ is not an exponent,} \\
\sum_{j=1}^{k} c_{s, j}^{-1}\mathcal{H}^{(s)}_j \left< V^{(s)}, \mathcal{H}^{(h-s)}_j \right> &\mbox{if } s \mbox{ is an exponent, }
\end{cases}
\end{aligned}
\end{equation}
where $c_{s , j} = \left< \mathcal{H}^{(s)}_j ,\mathcal{H}^{(h-s)}_j \right>$  and $k$ is the multiplicity of the exponent.
%---Recursion operators
From \eqref{Recurrence1} we can see that
\begin{equation}
\begin{aligned}
V^{(s-1)}_{\bot} &= \mbox{ad}_{J}^{-1}\left( i{\partial_x V^{(s)}}+ \big[ Q ,V^{(s)}_{\bot} \big] + \big[ Q ,V^{(s)}_{\|} \big]  \right), \\
i {\partial_x V^{(s)}_{\|}} &=- \big[ Q ,V^{(s)}_{\bot} \big]_{\|}.
\end{aligned}
\end{equation}
Integrating the second equation we get
\begin{equation}
\begin{aligned}
 V^{(s)}_{\|} &=
i \partial_x^{-1} \left[Q, V^{s}_{\bot} \right]_{\|} =
i \sum_{j=1}^{k} c_{s, j}^{-1} \mathcal{H}_j^{(s)}  \partial_x^{-1} \left< \left[Q ,V^{(s)}_{\bot} \right], \mathcal{H}_j^{(h-s)} \right>,
\end{aligned}
\end{equation}
where we have set any constants of integration to be equal to zero.
Thus the solution of the recurrent relations takes the form
\begin{equation}
\begin{aligned}
V^{(n-1)}&= -  \mbox{ad}_{J}^{-1}  \left[Q , K \right], \\
V^{(s-1)}_{\bot}  & =  \Lambda_s V^{(s)}_{\bot} .
\end{aligned}
\label{RecOp}
\end{equation}
There are two possible cases for the recursion operators $\Lambda_s$.
If $s$ is an exponent then
\begin{equation}
\begin{aligned}
\Lambda_{s} X &= \mbox{ad}_{J}^{-1}\Bigg( i{\partial_x X}+ \left[ Q , X \right] + i \sum_{j=1}^{k} c_{s, j}^{-1} \left[ Q ,\mathcal{H}^{(s)}_j \right] \partial_x^{-1} \left< \left[Q , X \right], \mathcal{H}^{(h-s)}_j \right>  \Bigg),
\end{aligned}
\end{equation}
where $k$ is the multiplicity of the exponent $s$.
If $s$ is not an exponent then
\begin{equation}\label{Lambda-s}\begin{aligned}
\Lambda_{s} X &= \mbox{ad}_{J}^{-1}\left( i{\partial_x X}+ \left[ Q , X \right] \right).
\end{aligned}
\end{equation}
%---Equations----
The hierarchy of soliton equations is given by
\begin{equation}\begin{split}
\label{Hierarchy}
{\partial_t Q} &= {\partial_x} \left( \Lambda_{2} \ldots \Lambda_{n-1} V^{(n-1)} \right) \\
&= \frac{\partial }{ \partial x }\left(  \Lambda_{2} \ldots \Lambda_{n-1} \ad_J^{-1}[K,Q] \right).
\end{split}\end{equation}
%----------------------------------------------
%
%---Hamiltonian formulation------------------
\section{Hamiltonian formulation}
\label{ham}
Every equation in the hierarchy \eqref{Hierarchy}  has infinitely many integrals of motion.
This is due to the underlying bi-Hamiltonian structure.
Every integral of motion can be viewed as a Hamiltonian with a properly chosen Poisson structure.
We will use the integral of motion given by \cite{VG-Ya-13,VG-Ya-14}
%---Integral of motion
\begin{equation}
I = \int_{{-\infty}}^{\infty} i c_{5,1}  \partial_x^{-1} \left< \left[ Q,  \Lambda_{0} V^{(0)} \right], \mathcal{H}^{(1)}_1 \right>{dx}.
\label{I1}
\end{equation}
The Hamiltonian $H$ is proportional to \eqref{I1}.
Hamilton's equations are
%----Hamilton equations-----------------
\begin{equation}
\partial_{t} q_i = \{ q_i, H \}
\label{HamEq}
\end{equation}
with a Poisson bracket given by
%-------Poisson bracket------------------
\begin{equation}
\{F,G\} = \int_{\mathbb{R}^2} \omega_{ij}(x,y) \frac{\delta{F}}{\delta{q_i}}\frac{\delta{G}}{\delta{q_j}} {dx} {dy},
\end{equation}
where we sum over repeating indexes.
The Poisson structure tensor is
%--------Poisson structure---------------
\begin{equation}
\omega_{ij}(x,y)= \frac{1}{2} \delta_{ij} \left( \partial_{x} \delta(x-y) - \partial_{y} \delta(x-y) \right).
\end{equation}
In this case \eqref{HamEq} reduces to
\begin{equation}
\partial_t q_i = \partial_x \frac{\delta{H}}{\delta{q_i}}.
\label{HamEq2}
\end{equation}

%----Particular Examples--------------------------------
\section{The one-parameter family of mKdV equations}
\label{example}
The first non-trivial member of the hierarchy is a set of mKdV equations which is obtained by \eqref{Hierarchy} and setting $n=3$.
Here, for the sake of completeness,  we will
%use an alternative approach
derive explicitly the potential of the $M$ operator, which in our particular case
is a qubic polynomial in $\lambda$.

As we mentioned above, the algebra $D_4$ is the only simple Lie algebra which has 3 as a double exponent. This means that the
element $K$ in (\ref{LaxPair}) involves two arbitrary parameters:
\begin{equation}
K = \frac{1}{2} a \mathcal{H}^{(3)}_1 + \frac{1}{6}b \mathcal{H}^{(3)}_4.
\end{equation}
The elements $V^{(k)}$ are given as a linear combination of the basis elements in $\mathfrak{g}^{(k)}$ with some coefficients $v_i^{(k)}$.

%V2
The coefficients of $V^{(2)}$ are obtained by solving the first equation from \eqref{RecOp}
\begin{equation}
\begin{aligned}
v_1^{(2)} &= 2 \omega a q_1, \\
v_2^{(2)} &= 0, \\
v_3^{(2)} &= - \omega (a+b)q_3, \\
v_4^{(2)} &= - \omega (a-b)q_4.
\end{aligned}\label{V2solved}
\end{equation}

%V1
From the second equation of \eqref{RecOp}, by setting $s=2$, for $V^{(1)}$ we have
\begin{equation}
\begin{aligned}
%v1_1
v_1^{(1)} &= - 2 a (\omega +1) \frac{\sqrt{3}}{3} \left( \partial_x q_1 -  \sqrt{3}q_4 q_3 + \sqrt{3}q_2 q_1 \right), \\
%v1_2
v_2^{(1)} &= 2 a q_1^2 - (a+b) q_3^2 - (a-b) q_4^2, \\
%v1_3
v_3^{(1)} &= (a+b) (\omega +1) \frac{\sqrt{3}}{3} \left(  \partial_x q_3 - \sqrt{3}q_1 q_4 - \sqrt{3}q_2 q_3 \right), \\
%v1_4
v_4^{(1)} &= (a-b) (\omega +1) \frac{\sqrt{3}}{3} \left( \partial_x q_4 - \sqrt{3}q_1 q_3 - \sqrt{3}q_2 q_4 \right), \\
%v1_5
v_5^{(1)} &= a q_1^2 - \frac{1}{2}(a+b)q_3^2 - \frac{1}{2}(a-b)q_4^2.
\end{aligned} \label{V1solved}
\end{equation}
Note that $v_5^{(1)}$ can formally be interpreted as "particle density" and
\begin{equation}
D= \int_{-\infty}^{\infty}\left( a q_1^2 - \frac{1}{2}(a+b)q_3^2 - \frac{1}{2}(a-b)q_4^2 \right) dx
\end{equation}
is an integral of motion.

%V0
Using \eqref{RecOp} and setting $s=1$ ,for $V^{(0)}$ we get
\begin{equation}
\begin{aligned}
%v0_1
v_1^{(0)} &= 2a(\partial^2_x q_1 - \sqrt{3} q_1 \partial_x q_2)
- \sqrt{3}\left( (3a+b)q_4 \partial_x q_3 + (3a-b)q_3 \partial_x q_4 \right) \\
 &- 3q_1(2 a q_2^2 - (a-b)q_3^2 - (a+b)q_4^2),\\ \\
%v0_2
v_2^{(0)} &= \sqrt{3} a \partial_x q_1^2 - \frac{\sqrt{3}}{2} (a+b) \partial_x q_3^2 - \frac{\sqrt{3}}{2} (a-b) \partial_x q_4^2  \\
&- 3 q_2 \left(2 a q_1^2 - (a+b)q_3^2 - (a-b) q_4^2 \right) ,\\ \\
%v0_3
v_3^{(0)} &= -(a+b)(\partial^2_x q_3 - \sqrt{3} q_3 \partial_x q_2)
+ \sqrt{3} \left( (3a+b) q_4 \partial_x q_1 + 2 b q_1 \partial_x q_4 \right) \\
&- 3 q_3 \left(2 a q_4^2 - (a-b) q_1^2 - (a+b) q_2^2 \right) ,\\ \\
%v0_4
v_4^{(0)} &= -(a-b)(\partial^2_x q_4 - \sqrt{3} q_4 \partial_x q_2)
+ \sqrt{3} \left( (3a-b) q_3 \partial_x q_1 - 2 b q_1 \partial_x q_3 \right) \\
&- 3 q_4 \left(2 a q_3^2 - (a-b) q_2^2 - (a+b) q_1^2 \right) .
\end{aligned}
\end{equation}

%---------Equations---------------------------
Finally, the equations obtained from \eqref{Hierarchy}, are
\begin{equation*}
\begin{aligned}
%eq1
\partial_t q_1 &= \partial_x \Big( 2a(\partial^2_x q_1 - \sqrt{3} q_1 \partial_x q_2)
- \sqrt{3}\left( (3a+b)q_4 \partial_x q_3 + (3a-b)q_3 \partial_x q_4 \right) \\
 &- 3q_1(2 a q_2^2 - (a-b)q_3^2 - (a+b)q_4^2) \Big),\\ \\
%eq2
\partial_t q_2 &= \partial_x \Big( \sqrt{3} a \partial_x q_1^2 - \frac{\sqrt{3}}{2} (a+b) \partial_x q_3^2 - \frac{\sqrt{3}}{2} (a-b) \partial_x q_4^2  \Big) \\
&- 3 q_2 \left(2 a q_1^2 - (a+b)q_3^2 - (a-b) q_4^2 \right) ,\\ \\
%eq3
\partial_t q_3 &= \partial_x \Big( -(a+b)(\partial^2_x q_3 - \sqrt{3} q_3 \partial_x q_2)
+ \sqrt{3} \left( (3a+b) q_4 \partial_x q_1 + 2 b q_1 \partial_x q_4 \right) \\
&- 3 q_3 \left(2 a q_4^2 - (a-b) q_1^2 - (a+b) q_2^2 \right) \Big) ,\\ \\
%eq4
\partial_t q_4 &= \partial_x \Big( -(a-b)(\partial^2_x q_4 - \sqrt{3} q_4 \partial_x q_2)
+ \sqrt{3} \left( (3a-b) q_3 \partial_x q_1 - 2 b q_1 \partial_x q_3 \right) \\
&- 3 q_4 \left(2 a q_3^2 - (a-b) q_2^2 - (a+b) q_1^2 \right) \Big). \\
\end{aligned}
\end{equation*}
%------Hamiltonian---------------------------
The density of the Hamiltonian given by \eqref{I1} is
 \begin{equation}\label{eq:H}
\begin{aligned}
\mathcal{H} &= \frac{1}{3} \left( 2 a q_1 \partial_x^2q_1 - (a+b) q_3 \partial_x^2 q_3 - (a-b) q_4 \partial_x^2 q_4 \right) \\
       &- \frac{1}{6} \left( 2a (\partial_x q_1)^2 - (a+b)(\partial_x q_3)^2 - (a-b)(\partial_x q_4)^2 \right) \\
       &+  2a \sqrt{3} \left( (q_4 q_3 + \frac{1}{3} q_2 q_1)\partial_x q_1  - \frac{1}{3} q_1^2 \partial_x q_2 \right) \\
       &- (a+b) \sqrt{3} \left( (q_4 q_1 - \frac{1}{3} q_2 q_3) \partial_x q_3 -\frac{1}{3}q_3^2 \partial_x q_2 \right) \\
       &- (a- b) \sqrt{3} \left( (q_3 q_1 - \frac{1}{3} q_2 q_4) \partial_x q_4 - \frac{1}{3}q_4^2 \partial_x q_2 \right) \\
       &- \frac{3}{2}\left( 2a (q_1^1 q_2^2 + q_4^2 q_32)-(a+b)(q_1^2 q_4^2 + q_2^2 q_3^2)-(a-b)(q_1^2 q_3^2  +q_2^2 q_4^2) \right)
\end{aligned}
\end{equation}
with Hamilton's equations given by \eqref{HamEq2}.

As we demonstrated above, the set of mKdV equations, as well as their Hamiltonian (\ref{eq:H})  involve two arbitrary parameters
$a$ and $b$. However, one of these parameters, say $a$, can always be put to 1 by re-scaling $t$ to $t/a$. That is why we have effectively
a one-parameter family of mKdV equations.

Note that for any other simple Lie algebra 3 is a single-valued exponent. As a result $K$ may involve only one parameter, which can always
be re-scaled to 1. Therefore the mKdV equations related to the other simple Lie algebras do not allow arbitrary parameters as their coefficients.

%------Spectral properties------------------------------
\section{Spectral properties of $L$ and fundamental analytic solutions}
\label{spectral}
%--------------Jost solutions----------------------------
Let us analyze the spectral properties of $L$ and outline the construction of its FAS. Basic tools in this analysis are the Jost solutions
\begin{equation}\label{eq:Jo}\begin{split}
\lim_{x\to -\infty}\phi_-(x,t,\lambda) e^{iJ\lambda x} =\openone, \qquad \lim_{x\to \infty}\phi_+(x,t,\lambda) e^{iJ\lambda x} =\openone.
\end{split}\end{equation}
Then the scattering matrix is introduced by
\begin{equation}\label{eq:T}\begin{split}
T(\lambda,t) = \hat{\phi}_+ ( x,t,\lambda) \phi_- ( x,t,\lambda),
\end{split}\end{equation}
where by "hat" we denote matrix inverse.
%---------------Volterra IE--------------------------------
Formally the Jost solutions must satisfy Volterra type integral equations.
If we introduce
\begin{equation}
\label{eq:xipm}
\xi_\pm (x,t,\lambda) =\phi_\pm (x,t,\lambda) e^{i\lambda Jx},
\end{equation}

then $\xi_\pm (x,t,\lambda)$ must satisfy

\begin{equation}
\label{eq:xipm'}
\begin{split}
\xi_+(x,t, \lambda) &= \openone + i \int_{\infty}^{x} dy \; e^{-i\lambda J(x-y)} Q(y,t) \xi_+ (y,t,\lambda)  e^{i\lambda J(x-y)}, \\
\xi_-(x,t, \lambda) &= \openone + i \int_{\infty}^{x} dy \; e^{-i\lambda J(x-y)} Q(y,t) \xi_- (y,t,\lambda)  e^{i\lambda J(x-y)}. \\
\end{split}
\end{equation}

However, in our case $J$ is complex-valued and a simple analysis shows that the Jost solutions exist only for potentials on compact support.
Beals and Coifman showed that any potential can be approximated to arbitrary accuracy with a potential on a finite support \cite{Beals-Coifman}.
This result was later generalized to potential with coefficients from any finite-dimensional simple Lie algebra \cite{GeYa*94}.
Reformulating the general results of \cite{GeYa*94,VG-Ya-14}, for the particular choice of the algebra $D_4\simeq so(8)$ we find that
%--------List-----------------------------
\begin{enumerate}
%----Continuous spectrum of L---------
\item The continuous spectrum of $L$ fills up the set of rays $l_\nu$, $\nu=0,\dots 11$ in the complex $\lambda$-plane for which
\begin{equation}
\label{eq:imla}
\begin{split}
\im \lambda \alpha(J) =0,
\end{split}
\end{equation}
where $\alpha$ is any root of $D_4$ and $J$ is given by (\ref{J}), see Figure \ref{fig:1}.
These rays are defined by
\begin{equation}
\label{eq:lnu}
\begin{split}
l_\nu = \arg \lambda = \frac{\pi \nu}{6}, \qquad \nu = 0,\dots 11.
\end{split}
\end{equation}
With each ray one can relate a subalgebra $\mathfrak{g}_\nu$ with root systems $\delta_\nu$ whose roots satisfy
\begin{equation}
\label{eq:gnu}
\begin{split}
\delta_\nu \equiv \{ \alpha \in \delta_\nu \quad \mbox{iff} \quad \im \lambda \alpha(J) =0 \quad \forall \lambda \in l_\nu\}.
\end{split}
\end{equation}
More specifically we have
\begin{equation}
\label{eq:g_nu}
\begin{aligned}
\delta_0 &\equiv \{ \pm (e_2+e_3), \pm (e_1- e_4), \pm (e_1+ e_4)\}, &\quad \delta_1 &\equiv \{ \pm (e_1+e_3)\}, \\
\delta_2 &\equiv \{ \pm (e_1-e_2), \pm (e_3 - e_4), \pm (e_3 + e_4)\}, &\quad \delta_3 &\equiv \{ \pm (e_2-e_3)\}, \\
\delta_4 &\equiv \{ \pm (e_1-e_3), \pm (e_2- e_4) , \pm (e_2 + e_4)\}, &\quad \delta_5 &\equiv \{ \pm (e_1+e_2)\}
\end{aligned}
\end{equation}
and $\delta_{\nu+6}\equiv \delta_\nu$, $\nu=0,\dots, 5$.

%-----------Figure SP-------------------------
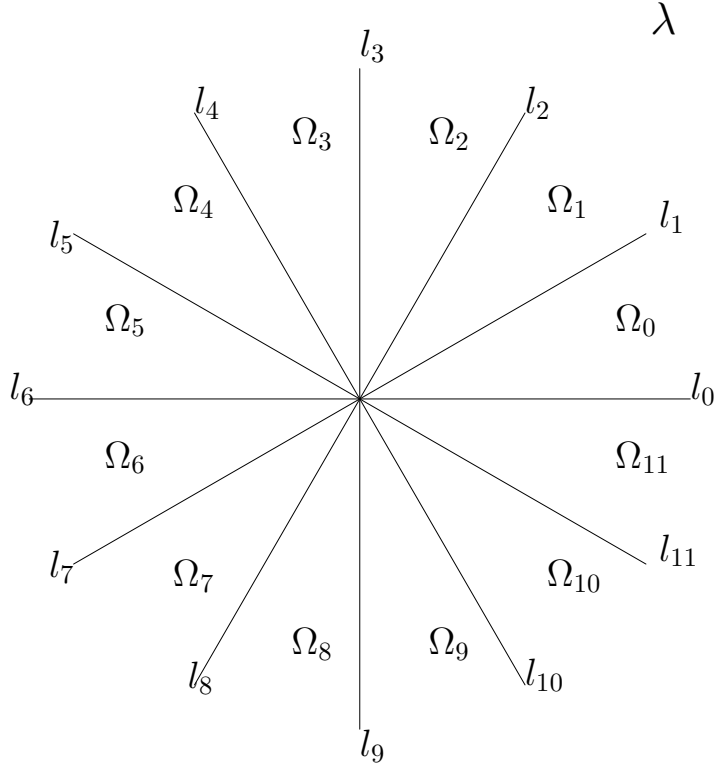
\begin{figure}[t]
\centering
\setlength{\unitlength}{0.2500pt}
\ifx\plotpoint\undefined\newsavebox{\plotpoint}\fi
\sbox{\plotpoint}{\rule[-0.175pt]{0.350pt}{0.350pt}}%
\special{em:linewidth 0.5pt}%
\begin{picture}(1440,1440)(0,0)

\put(1160.,1273){\Large $\lambda$}
%\put(220,720){\line(1,0){1000}} %\put(1260,720){\LARGE $l_1$}
%%%%%%%%%%%%%%%%%%%%%% PLUS \pi/10
\put(720,1240){\large $l_3$}
\put(1173,980.){\large $l_1$}
\put(1173,475.){\large $l_{11}$}
 \put(720,180.){\large $l_9$}
\put(250,450.){\large $l_7$}
\put(250,950.){\large $l_5$}

\put(970,1153){\large $l_2$}
\put(1220,720){\large $l_0$}
\put(970,287){\large $l_{10}$}
\put(460,282){\large $l_8$}
\put(190,720){\large $l_6$}
\put(470,1153){\large $l_4$}

\path(720,720)(470,287)
\path(720,720)(970,287)
\path(720,720)(1220,720)
\path(720,720)(970,1153)
\path(720,720)(470,1153)
\path(720,720)(220,720)

\path(720,720)(720,220)
\path(720,720)(1153,470)
\path(720,720)(1153,970)
\path(720,720)(720,1220)
\path(720,720)(287,970)
\path(720,720)(287,470)

\put( 1106.4, 823.53){\large $\Omega_0$}
\put(1002.8, 1002.8){\large $\Omega_1$}
\put( 823.53, 1106.4){\large $\Omega_2$}
\put( 616.47, 1106.4){\large $\Omega_3$}
\put(  437.16, 1002.8){\large $\Omega_4$}
\put(  333.63, 823.53){\large $\Omega_5$}
\put(  333.63, 616.47){\large $\Omega_6$}
\put(  437.16, 437.16){\large $\Omega_7$}
\put(  616.47, 333.63){\large $\Omega_8$}
\put(  823.53, 333.63){\large $\Omega_9$}
\put(  1002.8, 437.16){\large $\Omega_{10}$}
\put(  1106.4, 616.47){\large $\Omega_{11}$}
\end{picture}
\caption{The continuous spectrum of $L$ with $\mathbb{Z}_6$-symmetry fills up the rays $l_\nu$, $\nu=0,\dots, 11$; $\Omega_\nu$ are the analyticity regions of the FAS $\xi_\nu (x,t,\lambda)$.
\label{fig:1}}
\end{figure}
%--------------Regions of analyticity------------------------------
\item The regions of analyticity of the FAS $\xi_\nu (x,t,\lambda)$ are the sectors
\begin{equation}\label{eq:sect0}\begin{split}
\Omega_\nu \equiv \left\{ \frac{\pi \nu }{6} \leq \arg \lambda \leq \frac{\pi (\nu+1)}{6} \right\}.
\end{split}\end{equation}
They can be introduced as the solutions of the set of integral equations
\begin{equation}
\label{eq:fas1}
\begin{split}
\xi_{\nu,ij} (x,t,\lambda) = \delta_{ij} +  i \int_{s_{ij} \infty}^{x} dy \; e^{-i\lambda (J_i - J_j)(x-y)} (Q(y,t) \xi{\nu}(y,t,\lambda))_{ij},
\end{split}
\end{equation}
where $\lambda\in \Omega_\nu$ and  $s_{ij}$ take the values $\pm 1$, which are specific for each of the sectors $\Omega_\nu$, see Table \ref{tab:1}.

%-----------------Table of sings---------------------------------------
\begin{table}
  \centering
\begin{tabular}{|l|c|c|c|c|c|c|c|c|c|}
  \hline
  % after \\: \hline or \cline{col1-col2} \cline{col3-col4} ...
  & $(1,2) $ & $(1,3)$ & $ (1,4)$ & $(2,3)$ & $ (2,4)$ & $(3,4) $
  & $ (1,7)$ & $(1,6)$  & $ (2,6) $ \\
&  &  & $  (1,5)$ &  & $  (2,5)$ & $ (3,5) $  &  &   &  \\
  \hline
$  \Omega_0 $ & $ - $ & $ + $ & $ + $ & $ + $ & $ + $ & $ - $ & $ + $ & $ - $ & $ + $ \\
  $  \Omega_1 $ & $ - $ & $ + $ & $ + $ & $ + $ & $ + $ & $ - $ & $ + $ & $ + $ & $ +$ \\
   $ \Omega_2 $ & $ + $ & $ + $ & $ + $ & $ + $ & $ + $ & $ + $ & $ + $ & $ + $ & $ +$ \\
   $ \Omega_3 $ & $ + $ & $ + $ & $ + $ & $ - $ & $ + $ & $ + $ & $ + $ & $ + $ & $ + $ \\
   $ \Omega_4 $ & $ + $ & $ - $ & $ + $ & $ - $ & $ - $ & $ + $ & $ + $ & $ + $ & $ + $ \\
   $ \Omega_5 $ & $ + $ & $ - $ & $ + $ & $ - $ & $ - $ & $ + $ & $ - $ & $ + $ & $ + $ \\
     \hline
\end{tabular}
  \caption{The signs $s_{jk} = s_{9-k,9-j}$ in (\ref{eq:fas1}). We have listed only the values of
  $s_{jk}$ for $j<k$ because $s_{kj} =-s_{jk}$. Also the signs for $\Omega_{\nu+6}$ are opposite to the signs
  for $\Omega_\nu$, $\nu=0,\dots, 5$.  }\label{tab:1}
\end{table}

%-----------------Ordering---------------------------------
\item In each of the sectors $\Omega_\nu$ we introduce an ordering of the roots calling the root $\alpha$ $\nu$-positive (resp. $\nu$-negative) if $\im \lambda \alpha(J) >0$ (resp. $\im \lambda \alpha(J) <0$) for $\lambda\in \Omega_\nu$.
Thus the sets of positive roots of the subalgebras $\mathfrak{g}_\nu$ are
\begin{equation}
\label{eq:g_nup}
\begin{aligned}
\delta_0^+ &\equiv \{ e_2+e_3, e_1- e_4,   e_1+ e_4 \}, &\quad \delta_1^+ &\equiv \{  e_1+e_3\}, \\
\delta_2^+ &\equiv \{   e_1-e_2,  e_3 - e_4,   e_3 + e_4\}, &\quad \delta_3^+ &\equiv \{  e_2-e_3 \}, \\
\delta_4^+ &\equiv \{  e_1-e_3,   e_2- e_4 ,  e_2 + e_4\}, &\quad \delta_5^+ &\equiv \{   e_1+e_2\}.
\end{aligned}
\end{equation}
Note that the root systems $\delta_{1,3,5}$ are isomorphic to the root system of $sl(2)$, while each of the root systems $\delta_{0,2,4}$ is isomorphic to the direct sum of three $sl(2)$ algebras.
Indeed, it is enough to check that all the roots in $\delta_{0,2,4}^+$ are mutually orthogonal.

%---------Scattering data----------------------------------
\item As set of scattering data we introduce the limits of the FAS along both sides of all rays $l_\nu e^{\pm i0}$:
\begin{equation}
\label{eq:chi-as}
\begin{split}
\lim_{x\to -\infty} e^{-i\lambda Jx} \xi_\nu (x,t,\lambda ) e^{i\lambda Jx} &= S_\nu^+ (t,\lambda) , \\
\lim_{x\to \infty} e^{-i\lambda Jx} \xi_\nu (x,t,\lambda ) e^{i\lambda Jx}& = T_\nu^-(t,\lambda) D_\nu^+(\lambda) ,
\end{split}
\quad \forall \lambda \in l_\nu e^{+i0},
\end{equation}
%-----------------------------------------------------------------
\begin{equation}
\label{eq:lim}
\begin{split}
\lim_{x\to -\infty} e^{-i\lambda Jx} \xi_\nu (x,t,\lambda) e^{i\lambda Jx} &= S_{\nu+1}^- (t,\lambda) , \\
\lim_{x\to -\infty} e^{-i\lambda Jx} \xi_\nu (x,t,\lambda ) e^{i\lambda Jx} &= T_{\nu+1}^+(t,\lambda) D^-_{\nu+1}(\lambda),
\end{split}
\quad \forall \lambda \in l_{\nu+1} e^{-i0},
\end{equation}
where $S_\nu^\pm$,  $T_\nu^\pm$ and  $D_\nu^\pm$ are elements of the subgroup $\mathcal{G}_\nu$ of the form
%----------------------------------------------------------------
\begin{equation}
\label{eq:stdpm}
\begin{aligned}
S_\nu^\pm (t,\lambda)&= \exp \left( \sum_{\alpha\in \delta_\nu^+}^{} s^\pm_{\alpha,\nu}(t,\lambda) E_{\pm \alpha} \right), \\
T_\nu^\pm(t,\lambda) &= \exp \left( \sum_{\alpha\in \delta_\nu^+}^{} \tau^\pm_{\alpha,\nu}(t,\lambda) E_{\pm \alpha} \right), \\
 D_\nu^\pm (\lambda)&= \exp \left( \sum_{\alpha\in \delta_\nu^+}^{} d^\pm_{\alpha,\nu}(\lambda) H_{\alpha}\right)
\end{aligned}
\end{equation}
satisfying
\begin{equation}
\label{eq:tds}
T_\nu^- D_\nu^+ \hat{S_\nu^+} = T_\nu^+ D_\nu^- \hat{S_\nu^-} = T_\nu(\lambda,t), \qquad \lambda \in l_\nu.
\end{equation}
For potentials on compact support $T_\nu(\lambda,t)$ is the scattering matrix, which for $\lambda \in l_\nu$ is also taking values in the subgroup $\mathcal{G}_\nu$. Then (\ref{eq:tds}) is the Gauss decomposition of  $T_\nu(\lambda,t)$.
We only need to add, that the functions $D_\nu^+$ and $D_{\nu+1}^-$ are analytic in the sector $\Omega_\nu$.

%---------------Reimann-Hilbert problem------------------------------
\item
The FAS $\xi_\nu (x,t,\lambda)$ satisfy a RHP:
\begin{equation}
\label{eq:rhp1}
\begin{split}
\xi_{\nu+1} (x,t,\lambda) = \xi_\nu (x,\lambda)G_\nu (x,\lambda),
\qquad G_\nu (x,\lambda) = e^{-i\lambda Jx}\hat{S}^-_{\nu+1} S^+_{\nu+1} (\lambda) e^{i\lambda Jx},
\end{split}
\end{equation}
which allows canonical normalization:
\begin{equation}
\label{eq:rhp2}
 \lim_{\lambda\to\infty} \xi_\nu (x,\lambda) =\openone.
\end{equation}
From Zakharov-Shabat theorem \cite{ZaSha}
the solution of the RHP (\ref{eq:rhp1}) with canonical normalization is a FAS of the system
\begin{equation}
\label{eq:xinu}
 i \frac{\partial \xi_\nu}{ \partial x } + Q(x)\xi_\nu (x,\lambda) -\lambda [J, \xi_\nu (x,\lambda)]=0,
\end{equation}
which means that $\chi_\nu (x,\lambda) =\xi_\nu (x,\lambda) e^{-i\lambda Jx}$ are FAS of the Lax operator $L$.

\end{enumerate}

%--------------------------------------------------
A basic step in applying the ISM is the analysis of the mapping between the space of admissible potentials
$\mathcal{M}$ and the space of scattering data. In view of Zakharov-Shabat theorem it is natural to expect that
the minimal sets of scattering data are determined by the set of $S_\nu^\pm (t,\lambda)$.
In fact we can prove the following lemma
%--------Lemma 1-------------------------------
\begin{lemma}
\label{lem:1}
Assume that the solutions  $\xi_\nu(x,t,\lambda)$ of the RHP (\ref{eq:rhp1}) are regular, i.e. they have no zeros or
singularities in their regions of analyticity.
Then there are two minimal sets of scattering data determining uniquely both the scattering matrix $T(t,\lambda)$ and the
potential $Q(x,t)$
\begin{equation}
\label{eq:T12}
\begin{split}
\mathcal{T}_1 \equiv \{  s_{0,\alpha}^+ (t,\lambda) , \quad \alpha\in\delta_0^+,   \quad  s_{1,\alpha}^\pm (t,\lambda) , \quad \alpha\in\delta_1^+, \quad s_{2,\alpha}^- (t,\lambda) , \quad \alpha\in\delta_2^-, \}, \\
\mathcal{T}_2 \equiv \{  s_{0,\alpha}^+ (t,\lambda) , \quad \alpha\in\delta_0^+,   \quad  s_{1,\alpha}^\pm (t,\lambda) , \quad \alpha\in\delta_1^+, \quad s_{2,\alpha}^- (t,\lambda) , \quad \alpha\in\delta_2^-, \}.
\end{split}
\end{equation}
\end{lemma}

%---------- Proof of lemma 1 ------------------------------
\begin{proof}
It is well known, see \cite{ContM,VG-Ya-13,VG-Ya-14}
that regularity of the solutions of RHP ensures that relevant Lax operator
$L$ has no discrete eigenvalues.
More precisely, this means that the functions $D_\nu^\pm(\lambda)$ have no zeros or singularities.

First we will prove that each of the sets $\mathcal{T}_i$, $i=1,2$ determines uniquely the full set of scattering data.
Assume we are given the set $\mathcal{T}_1$.
Looking at the first equation from (\ref{eq:stdpm}),  $\mathcal{T}_1$ uniquely determines the matrices $S_0^\pm(\lambda)$ for $\lambda\in l_0 e^{\pm i0}$ and $S_1^\pm(\lambda)$ for $\lambda\in l_1 e^{\pm i0}$.
 Using the $\mathbb{Z}_6$ reduction condition we can determine the matrices  $S_\nu^\pm(\lambda)$ for $\lambda\in l_\nu e^{\pm i0}$ along all other rays by the relations
\begin{equation}
\label{eq:Snu}
\begin{aligned}
S_{2\nu}^\pm(\lambda) &= c^\nu S_0^\pm(\omega^{-\nu} \lambda) c^{-\nu}, \qquad \lambda \in l_\nu e^{\pm i0}, \\
S_{2\nu+1}^\pm(\lambda) &= c^\nu S_1^\pm(\omega^{-\nu} \lambda) c^{-\nu}, \qquad \lambda \in l_{\nu+1} e^{\pm i0}.
\end{aligned}
\end{equation}
The next step is to demonstrate that by knowing all $S_\nu^\pm(\lambda)$ we can recover $D_\nu^\pm (\lambda)$.
To this end we make use of the analyticity properties of $D_\nu^\pm(\lambda)$ and  (\ref{eq:tds}).
Indeed, let us sandwich (\ref{eq:tds}) by $\langle \omega_{\nu,j}^\pm| \cdots
|\omega_{\nu,j}^\pm \rangle$, where $\omega_{\nu,j}^+$ is the $j$-th fundamental weight of the subalgebra $\mathfrak{g}_\nu$ evaluated with the ordering in $\Omega_\nu$.
 Obviously $\omega_{\nu,j}^+$ are the highest weight vectors in the corresponding fundamental representation of $\mathfrak{g}_\nu$.
Analogously $\omega_{\nu,j}^-$ are the lowest weight vectors in the same  fundamental representation of $\mathfrak{g}_\nu$.
We also need the following properties of the weight vectors $\langle  \omega_{\nu,j}^\pm |$, $| \omega_{\nu,j}^\pm \rangle $ and the factors $S_{\nu}^\pm(\lambda)$, $T_{\nu}^\pm(\lambda)$, $D_{\nu}^\pm(\lambda)$:
\begin{equation}
\label{eq:Snupm}
\begin{aligned}
\langle  \omega_{\nu,j}^+ | S_{\nu}^-(\lambda) &= \langle  \omega_{\nu,j}^+ | T_{\nu}^-(\lambda) =\langle  \omega_{\nu,j}^+ | , &\quad
\langle  \omega_{\nu,j}^-| S_{\nu}^+(\lambda) &= \langle  \omega_{\nu,j}^- | T_{\nu}^+(\lambda) =\langle  \omega_{\nu,j}^- | , \\
 S_{\nu}^+(\lambda) |\omega_{\nu,j}^+  \rangle   &=  T_{\nu}^+(\lambda) |\omega_{\nu,j}^+  \rangle  =|\omega_{\nu,j}^+  \rangle  , &\quad
 S_{\nu}^-(\lambda) |\omega_{\nu,j}^-  \rangle   &=  T_{\nu}^-(\lambda) |\omega_{\nu,j}^-  \rangle  =|\omega_{\nu,j}^-  \rangle
\end{aligned}
\end{equation}
and
\begin{equation}
\label{eq:Dnu}
\begin{split}
\langle  \omega_{\nu,j}^+ | D_{\nu}^+(\lambda) |\omega_{\nu,j}^+  \rangle   &= e^{d_{\nu,\alpha_j}^+}, \qquad
\langle  \omega_{\nu,j}^- | D_{\nu}^-(\lambda) |\omega_{\nu,j}^-  \rangle   = e^{-d_{\nu,\alpha_j}^-}.
\end{split}
\end{equation}
Thus from (\ref{eq:tds}) we obtain that
\begin{equation}
\label{eq:dnupm}
\begin{aligned}
d_{\nu,\alpha_j}^+ - d_{\nu,\alpha_j}^- &= \ln \langle  \omega_{\nu,j}^- | \hat{S}_{\nu}^-(\lambda) S_{\nu}^+(\lambda) |  \omega_{\nu,j}^- \rangle \\
&= - \ln \langle  \omega_{\nu,j}^+ | \hat{T}_{\nu}^+(\lambda) T_{\nu}^-(\lambda) |  \omega_{\nu,j}^+ \rangle .
\end{aligned}
\end{equation}
The functions $d_{\nu,\alpha_j}^\pm$ can be recovered uniquely from their analyticity properties and from the jumps (\ref{eq:dnupm}) along the rays $l_\nu$, which in turn allows us to recover the factors $D_\nu^\pm (\lambda)$.

As a result, we can rewrite (\ref{eq:tds}) in the form
\begin{equation}
\label{eq:tds2}
\begin{split}
D_\nu^- \hat{S_\nu^-}S_\nu^+  \hat{D}_\nu^+= \hat{ T}_\nu^+ (\lambda,t) T_\nu^- ,  \qquad   \lambda \in l_\nu.
\end{split}
\end{equation}
Now the factors $T_\nu^\pm (\lambda)$ are recovered as the unique Gauss factors of the right hand side of (\ref{eq:tds2}).

The proof that $\mathcal{T}_2$ determines uniquely all the scattering data is analogous.

Finally,  the corresponding potential is reconstructed from \cite{NMPZ,ContM,GeYa*94}
\begin{equation}
\label{eq:Qxt}
\begin{split}
Q(x,t) = \lim_{\lambda\to\infty} \lambda \left( J - \xi_\nu(x,t,\lambda) J \hat{\xi}_\nu(x,t,\lambda) \right),
\end{split}
\end{equation}
where $\xi_\nu(x,t,\lambda) $ is the unique regular solution of the RHP (\ref{eq:rhp1}).

\end{proof}

%--------Time dependence of the scattering data
\subsection{Time dependence of the scattering data}
\label{sdata}
The Lax representation ensures that $L$ and $M$ allow the same set of fundamental solutions.
As such we shall choose the FAS. A bit more generally, we can assume that
\begin{equation}
\label{eq:LM2}
L\chi_\nu(x,t,\lambda)=0, \qquad  M\chi_\nu(x,t,\lambda)-\chi_\nu(x,t,\lambda) \Gamma(\lambda) =0,
\end{equation}
where $\phi_+(x,t,\lambda) = \xi_+(x,t,\lambda)e^{i\lambda Jx}$ and $\Gamma(\lambda)$ is a constant matrix.
We shall fix up $\Gamma(\lambda)$ in such a way, that the definition of the Jost solutions are $t$-independent.

Let $\lambda\in l_\nu e^{i0}$ and let us  calculate the limit
\begin{equation}
\label{eq:limM}
\begin{split}
&\lim_{x \to -\infty} \left( M\chi_{\nu}(x,t,\lambda)\right) e^{i\lambda J x}  \\
& \qquad = \lim_{x\to -\infty}\left[\left(i\frac{\partial}{\partial t}+ \sum_{p=0}^{m-1}\lambda^p V_{p} -\lambda^{m} K \right) \chi_\nu(x,t,\lambda)
-  \chi_\nu(x,t,\lambda)  \Gamma(\lambda) \right] e^{i\lambda J x} \\
&\qquad = i \frac{\partial S_\nu^+}{ \partial t } - \lambda^m K S_\nu^+ (t,\lambda) -
\lim_{x\to -\infty} S_\nu^+ (t,\lambda)e^{-i\lambda J x}  \Gamma(\lambda) e^{i\lambda J x}, \qquad \lambda \in l_\nu e^{+i0},
\end{split}
\end{equation}
where we took into account that all $V_k(x,t)$ vanish for $x\to\pm\infty$. Thus we obtain that
\begin{equation}
\label{eq:Gam0}
i \frac{\partial S_\nu^+}{ \partial t } - \lambda^m K S_\nu^+ (t,\lambda) -
\lim_{x\to -\infty} e^{-i\lambda J x} S_\nu^+ (t,\lambda)\Gamma(\lambda) e^{i\lambda J x} =0.
\end{equation}
First we consider the diagonal part of the above equation taking into account that from (\ref{eq:stdpm}) it follows
that all diagonal elements of $S_\nu^+(t,\lambda)$ are equal to 1. This immediately gives
\begin{equation}
\label{eq:Gam}
\begin{split}
\lambda^m K + \lim_{x\to -\infty} e^{-i\lambda J x} \Gamma(\lambda) e^{i\lambda J x} =\lambda^m K + \Gamma(\lambda) =0,
\end{split}
\end{equation}
i.e. $\Gamma(\lambda) =-\lambda^m K$. Then (\ref{eq:Gam0}) becomes
\begin{equation}
\label{eq:Snut}
 i \frac{\partial S_\nu^+}{ \partial t } - \lambda^m [K, S_\nu^+ (t,\lambda) ]=0, \qquad \lambda \in l_\nu e^{+i0}.
\end{equation}

Next we calculate similar limits for $x\to \infty$:
\begin{equation}
\label{eq:Mphim}
\begin{split}
&\lim_{x \to \infty} \left( M\chi_{\nu}(x,t,\lambda)\right) e^{i\lambda J x}  \\
& \qquad = \lim_{x\to \infty}\left[\left(i\frac{\partial}{\partial t}+ \sum_{p=0}^{m-1}\lambda^p V_{p} -\lambda^{m} K \right) \chi_\nu(x,t,\lambda)
-  \chi_\nu(x,t,\lambda) \Gamma(\lambda) \right] e^{i\lambda J x} \\
&\qquad =i \frac{\partial T_\nu^- D_\nu^+}{ \partial t } - \lambda^m [K, T_\nu^-  (t,\lambda) D_\nu^+ (\lambda) ]=0.
\end{split}
\end{equation}
Again we consider first the diagonal part of (\ref{eq:Mphim}) and remember that all diagonal elements of the lower
triangular matrix $T_\nu^-(t,\lambda)$ are equal to 1. We obtain that
\begin{equation}
\label{eq:dDt}
\begin{split}
i \frac{\partial D_\nu^+}{ \partial t } =0,
\end{split}
\end{equation}
i.e. the matrix elements of $D_\nu^+(\lambda)$ are generating functionals of the integrals of motion of the NLEE.
Finally, the lower triangular part of (\ref{eq:Mphim}) gives
\begin{equation}
\label{eq:dTm}
i \frac{\partial T_\nu^-}{ \partial t } - \lambda^m [K, T_\nu^- (t,\lambda) ]=0, \qquad \lambda \in l_\nu e^{+i0}.
\end{equation}

Analogously we can set $\lambda\in l_\nu e^{-i0}$ and evaluate the limits for $x\to\pm\infty$.
Therefore, if $Q(x,t)$ satisfies the corresponding equations then the time evolution of the scattering data will be given by
\begin{equation}
\label{eq:dscatdt}
\begin{aligned}
i \frac{\partial S_\nu^+}{ \partial t } &- \lambda^m [K, S_\nu^+ (t,\lambda) ]=0,
&i \frac{\partial T_\nu^-}{ \partial t } &- \lambda^m [K, T_\nu^- (t,\lambda) ]=0,  \quad \lambda \in l_\nu e^{+i0}, \\
i \frac{\partial S_\nu^-}{ \partial t } &- \lambda^m [K, S_\nu^- (t,\lambda) ]=0,
&i \frac{\partial T_\nu^+}{ \partial t } &- \lambda^m [K, T_\nu^ +(t,\lambda) ]=0,  \quad \lambda \in l_\nu e^{-i0}
\end{aligned}
\end{equation}
and the diagonal factors $D_\nu^\pm (\lambda)$ are $t$-independent.
In other words $d_{\alpha,\nu}^\pm (\lambda)$ are also time-independent and are more convenient as generating functionals
of conservation laws.
The reason for this is, that if we consider the expansion of, say $d_{\alpha,\nu}^+ (\lambda)$ over the inverse powers of $\lambda$, that is
\begin{equation}\label{eq:dnu}\begin{split}
d_{\alpha,\nu}^+ (\lambda) = \sum_{ p=1}^{\infty} C_{\alpha,\nu }^{(p)} \lambda^{-p} ,
\end{split}\end{equation}
then their densities will be local, i.e. they will be expressed  in terms of the coefficients of $Q(x,t)$ and their $x$-derivatives, see \cite{Mikhailov,VG-Ya-13}.

The solution of the equations (\ref{eq:dscatdt}) is given by
\begin{equation}
S_{\nu,kj}^\pm (t,\lambda)= e^{i\lambda^{m} (K_k - K_j)t} S_{\nu,ij}^\pm(\lambda,0), \qquad
T_{\nu,kj}^\pm (t,\lambda)= e^{i\lambda^{m} (K_k - K_j)t} T_{\nu,ij}^\pm(\lambda,0).
\end{equation}

Thus $ S_{\nu}^\pm(\lambda,0)$ (resp. $ T_{\nu}^\pm(\lambda,0)$) for $\nu =0,1$ can be viewed  as the Cauchy data for the initial conditions.
Therefore solving the corresponding equations reduces to solving the direct and the inverse scattering problem for the Lax operator $L$, see \cite{Mikhailov,VG-Ya-13}.

%-----------------------------------------------------------------
\section{Discussion and conclusions}
\label{end}
In this paper we investigated some aspects of the deep relation between integrable equations and infinite-dimensional Lie algebras.
More concretely, we studied  the hierarchy corresponding in the scheme of Drinfeld and Sokolov to the untwested affine Kac-Moody Lie algebra $D_4^{(1)}$.

The NLEE related to $D_4^{(1)}$ deserve special attention due to the unique properties of $D_4^{(1)}$.
This is the only simple Lie algebra whose
Dynkin diagram has an $S_3$ symmetry and which has  3 as a double-valued exponent. We formally constructed the hierarchy of soliton equations
using the recursion operators. Thus the general theory of the recursion operators with deep reductions \cite{VG-Ya-13,VG-Ya-14} has been
adapted to the specific case of $D_4$.

Special attention is paid to the first nontrivial member of the hierarchy which is an one-parameter family of mKdV equations. Note that this
family of mKdV equations is unique! Among the simple Lie algebras only $D_{2r}^{(1)}$ allow double-valued exponents equal to $2r-1$. Thus with
each such algebra one can relate an one-parameter family of mKdV-like equations. These will be systems of $2r$ evolution equations for $2r$
independent functions whose leading $x$-derivative terms will be of order $2r-1$. In particular, with $D_6^{(1)}$ one could relate a system
of 6 evolution equations for 6 independent functions whose leading $x$-derivative will be of order 5.
All these equations have an infinite number of integrals of motion and a natural Hamiltonian formulation.
The properties of their Hamiltonian hierarchies will analyzed elsewhere.

We also outlined  the spectral properties of the Lax operator $L$ which has an in-built $\mathbb{Z}_6$-reduction
of Mikhailov type \cite{Mikhailov}. In particular, we constructed the FAS of the system and the associated
RHP. We also formulated two minimal sets of scattering data of $L$ from which one can recover uniquely both the
scattering matrix $T(\lambda,t)$ and the potential $Q(x,t)$ of $L$.

%What will be done
A natural extension of these results is to study the soliton equations associated with the twisted affine Kac-Moody algebras.
The corresponding integrable hierarchies related to the twisted affine Kac-Moody algebras
$A^{(2)}_{2r}$, $A^{(2)}_{2r-1}$ and $D^{(2)}_{4}$ are under investigations and have been presented elsewere
\cite{GMSV5,GMSV6,GMSV7}.

Another continuation in this direction involves constructing the soliton solutions of the obtained integrable hierarchies
using, for example, the dressing method \cite{ZakharovShabat2,Mikhailov} and analyzing their properties.

% ------------------Acknowledgments---------------------------

\section*{Acknowledgements}
We are grateful to Laslo Feh\'{e}r for the very useful discussions.
The work is partially supported by the ICTP - SEENET-MTP project PRJ-09.
One of us (VSG) is grateful to professors A.S. Sorin, A. V. Mikhailov and V. V. Sokolov for useful discussions during his visit to
JINR, Dubna, Russia under project 01-3-1116-2014/2018.

%--------------------- Appendix  ------------------------------------
%-----------------------------------------------------------
\section*{Appendix}
\appendix
%-------D4--------------------------------------------------------------
\section{The simple Lie algebra $D_4$}

%Cartan-Weyl basis
Let $\mathfrak{g}$ be a simple Lie algebra.
A Cartan-Weyl basis in $\mathfrak{g}$ is a system of generators $H_i , E_{\alpha_i}, E_{- \alpha_i},1 \leq i \leq r$ such that
\begin{equation}
\begin{aligned}
&\big[ H_i, H_j \big]= 0, \\
&\big[  E_{\alpha_i} ,  E_{- \alpha_j} \big] = \delta_{ij} H_i, \\
&\big[  H_i ,  E_{\alpha_j} \big] = C_{ij} H_i, \\
&\big[  H_i ,  E_{- \alpha_j} \big] = -C_{ij} H_i.
\end{aligned}
\end{equation}
%Cartan matrix
The matrix $(C_{ij})$ is non-degenerate and is called the Cartan matrix of $\mathfrak{g}$. The number $r$ is called the rank of the algebra. The generators $H_i$ form an Abelian subalgebra called the Cartan subalgebra.
Every simple Lie algebra is uniquely determined by its Cartan matrix and can also be represented by a Dynkin diagram \cite{Carter}.% For details on the classification of simple Lie algebras see \cite{SLA}.

%Properties of D4(1), Redefinition of orthogonality
The simple Lie algebra $D_4 \equiv \mathfrak{so}(8)$ is usually represented by a $8\times 8$ antisymmetric matrices.
In this realization the Cartan subalgebra is not diagonal, so we will use a realization for which every $X \in D_4$ satisfies
\begin{equation}
SX+(SX)^T=0,
\end{equation}
where the matrix $S$ is given by
\begin{equation}
S=
\begin{pmatrix}
0 & 0 & 0 & 0 & 0 & 0 & 0 & 1 \\
0 & 0 & 0 & 0 & 0 & 0 & -1 & 0 \\
0 & 0 & 0 & 0 & 0 & 1 & 0 & 0  \\
0 & 0 & 0 & 0 & -1 & 0 & 0 & 0 \\
0 & 0 & 0 & -1 & 0 & 0 & 0 & 0 \\
0 & 0 & 1 & 0 & 0 & 0 & 0 & 0 \\
0 & -1 & 0 & 0 & 0 & 0 & 0 & 0 \\
1 & 0 & 0 & 0 & 0 & 0 & 0 & 0
\end{pmatrix}.\label{S}
\end{equation}
This way the Cartan subalgebra is given by diagonal matrices.

%Cartan-Weyl basis of D_4
The Cartan-Weyl basis of $D_4$ is given by
\begin{equation}
\begin{aligned}
H_i &= e_{ii} - e_{9-i,9-i},  \quad 1 \leq i \leq 4, \\
E_{\alpha_j} &= e_{j, j+1}+e_{8-j,9-j} \quad 1 \leq j \leq 3, \\
E_{\alpha_4} &= e_{3,5}+ e_{4,6}, \quad E_{-\alpha_j}=(E_{\alpha_j})^T,
\end{aligned}
\end{equation}
where by $e_{ij}$ we denote a matrix that has a one at the $i-th$ row and $j-th$ column and is zero everywhere else.

%----Coxeter automorphism---------------------------------------
\section{Dihedral realization of the Coxeter automorphism}

Let $e_i$ be an orthonormal basis in the root space of $D_4$ and let the simple roots be of the form
%---Simple roots------------------------------------------------
\begin{equation}
\begin{aligned}
\alpha_1 = e_1 - e_2, \qquad \alpha_2 = e_2 - e_3, \qquad \alpha_3 = e_3 - e_4, \qquad \alpha_4 = e_3 + e_4.
\end{aligned}
\end{equation}
%------Dihedral realization-----------------------------------
In the dihedral realization the Coxeter automorphism is an element of the Weyl group of the form
\begin{equation}
C=w_1 w_2
\end{equation}
with
\begin{equation}
w_1 = \prod_{\alpha \in A_1} S_{\alpha}, \qquad  w_2 = \prod_{\beta \in A_2} S_{\beta},
\end{equation}
where $A_1$ and $A_2$ are subsets of the simple roots $A= \{ \alpha_1, \alpha_2, \alpha_3, \alpha_4 \}$, such that
\begin{equation}
A_1 = \{ \alpha_1, \alpha_3, \alpha_4 \},  \: A_2 =  \{ \alpha_2 \}.
\end{equation}
Here $S_\alpha$ is the reflection with respect to the simple root $\alpha$.
%---------Coxeter element-----------------------------------
For the Coxeter element we get
\begin{equation}
C= S_{\alpha_1} S_{\alpha_3} S_{\alpha_4} S_{\alpha_2}.
\end{equation}
The action of $C$ on the basis vectors is
\begin{equation}
\begin{aligned}
&C: e_1 \to e_2 \to -e_3 \to -e_1 \to -e_2 \to e_3, \\
&C: e_4 \to -e_4.
\end{aligned}
\label{orbits1}
\end{equation}
%-------------------------------------------------------------------
 Note that the eigenvalues of C acting in root space are $1,3,5,3$ and are called exponents of the algebra.
It is obvious that $C$ divides the root space into orbits. It is not difficult to check that each orbit contains only one simple root and the orbits are not intersecting.
When constructing the basis \eqref{basis1} we will select one element from each orbit, namely, the simple roots $\alpha_i$.

In the algebra $C$ is realized as an inner automorphism , i.e. a similarity transformation
\begin{equation}
C(X) = c X c^{-1} \label{sim1} , \qquad X \in D_4^{(1)},
\end{equation}
where $c$ belongs to the corresponding group. In order to find $c$ we will look at the typical (defining) representation of $D_4$.
%-------Coxeter - Typical representation of D4------------------
From \eqref{orbits1} we can infer that
\begin{equation}
c
\begin{pmatrix}
e_1 \\
e_2 \\
e_3 \\
e_4 \\
-e_4 \\
-e_3 \\
-e_2 \\
-e_1
\end{pmatrix}
=
\begin{pmatrix}
a_1 e_2 \\
a_2 (-e_3) \\
a_3 e_1 \\
a_4 (-e_4) \\
a_5 (e_4 )\\
a_6 (-e_1)\\
a_7 (-e_3) \\
a_8 (-e_2)
\end{pmatrix}, \label{trep1}
\end{equation}
for some coefficients $a_i$.
%-------------------------------------------------------------------
The following must also be true
\begin{equation}
c S c^T S = \mathbb{I}     ,\qquad c^6= \pm \mathbb{I}, \label{req1}
\end{equation}
i.e. the matrix $c$ must be a group element and it must give the Coxeter automorphism (the Coxeter number for $D_4^{(1)}$ is 6).
%-------------------------------------------------------------------
From \eqref{trep1} and \eqref{req1} we find that one possible choise of $c$ is given by
\begin{equation}
c =
\begin{pmatrix}
0 & -1 & 0 & 0 & 0 & 0 & 0 & 0 \\
0 & 0 & 0 & 0 & 0 & -1 & 0 & 0 \\
1 & 0 & 0 & 0 & 0 & 0 & 0 & 0  \\
0 & 0 & 0 & 0 & -1 & 0 & 0 & 0 \\
0 & 0 & 0 & -1 & 0 & 0 & 0 & 0 \\
0 & 0 & 0 & 0 & 0 & 0 & 0 & 1 \\
0 & 0 & 1 & 0 & 0 & 0 & 0 & 0 \\
0 & 0 & 0 & 0 & 0 & 0 & 1 & 0
\end{pmatrix}.
\label{coxeter1}
\end{equation}
%-------------------------------------------------------------------------
%\begin{align}
%&C: \alpha_1 \to \alpha_2 + \alpha_3 + \alpha_4 \to \alpha_1 + \alpha_2 \to -\alpha_1 \to -(\alpha_2 + \alpha_3 + \alpha_4) \to -(\alpha_1), \\
%&C: \alpha_2 \to -(\alpha_1 + \alpha_2 + \alpha_3 + \alpha_4) \to -(\alpha_1 + 2\alpha_2 + \alpha_3 + \alpha_4) \to -\alpha_2 \to \alpha_1 + \alpha_2 + \alpha_3 + \alpha_4 \to \alpha_1 + 2\alpha_2 + \alpha_3 + \alpha_4, \\
%&C: \alpha_3 \to \alpha_1 + \alpha_2 + \alpha_4 \to \alpha_2 + \alpha_3 \to - \alpha_3 \to -(\alpha_1 + \alpha_2 + \alpha_4) \to -(\alpha_2 + \alpha_3), \\
%&C: \alpha_4 \to \alpha_1 + \alpha_2 + \alpha_3 \to \alpha_2 + \alpha_4 \to -\alpha_4 \to -(\alpha_1 + \alpha_2 + \alpha_3) \to -(\alpha_2 + \alpha_4)
%\end{align}
%
%---Basis-----------------------------------------------------------------
%
%Make all matrices have the same width
\newcolumntype{C}[1]{>{\centering\arraybackslash$}m{#1}<{$}}
\newlength{\mycolwd}
\settowidth{\mycolwd}{$-\omega^2 x_4$}
\section{Basis in $D_4^{(1)}$}
Every element of $D_4^{(1)}$ is of the form
\begin{equation}
X=\sum_{k= n}^{m} X^{(k)} \lambda^{k}, \quad n,m \in \mathbb{Z},
\end{equation}
where $X^{(k)} \in \mathfrak{g}^{(k \mbox{ mod } 6)}$.
For the elements $X^{(k)}$ we have
\begin{equation}
 X^{(k)}=\sum_{i=1}^{4} x_i  \mathcal{E}_{i}^{(k)} + \sum_{j=1}^{s}  x_{4+s} \mathcal{H}^{(k)}_{s},
\end{equation}
 where $s$ is the multiplicity of the exponent $k$. The explicit form of $X^{(k)}$ is
%----g(0)------------------------------------------------------------------
\begin{equation}
X^{(0)}=
{\footnotesize
\left (
\begin{array}{*{9}{@{}C{\mycolwd}@{}}}
0 & x_1 & x_1 & -x_4 & -x_3 & x_2 & -x_2 & 0 \\
-x_1 & 0 & x_2 & -x_3 & -x_4 & x_1 & 0 & -x_2 \\
-x_1 & -x_2 & 0 & x_3 & x_4 & 0 & x_1 & -x_2  \\
x_4 & x_3 & -x_3 & 0 & 0 & x_4 & x_4 & -x_3 \\
x_3 & x_4 & -x_4 & 0 & 0 & x_3 & x_3 & -x_4 \\
-x_2 & -x_1 & 0 & -x_4 & -x_3 & 0 & x_2 & -x_1 \\
x_2 & 0 & -x_1 & -x_4 & -x_3 & -x_2 & 0 & x_1 \\
0 & x_2 & x_2 & x_3 & x_4 & x_1 & -x_1 & 0
\end{array}
\right )
},\label{g0}
\end{equation}
%----g(1)------------------------------------------------------------------
\begin{equation}
X^{(1)}=
{\footnotesize
\left (
\begin{array}{*{9}{@{}C{\mycolwd}@{}}}
2 x_5 & x_1 & \omega^2 x_1 & -\omega x_4 & -\omega x_3 & -\omega x_2 & \omega^2 x_2 & 0 \\
x_1 & 2 \omega x_5 & x_2 & -\omega^2 x_3 & -\omega^2 x_4 & \omega x_1 & 0 & \omega^2 x_2 \\
\omega^2 x_1 & x_2 & -2 \omega^2 x_5 & x_3 & x_4 & 0 & \omega x_1 &  \omega x_2  \\
-\omega x_4 & -\omega^2 x_3 & x_3 & 0 & 0 & x_4 & \omega^2 x_4 & -\omega x_3 \\
-\omega x_3 & -\omega^2 x_4 & x_4 & 0 & 0 & x_3 & \omega^2 x_3 & -\omega x_4 \\
-\omega x_2 & \omega x_1 & 0 & x_4 & x_3 & 2\omega^2 x_5 & x_2 & -\omega^2 x_1 \\
\omega^2 x_2 & 0 & \omega x_1 & \omega^2 x_4 & \omega^2 x_3 & x_2 & -2 \omega x_5 & x_1 \\
0 & \omega^2 x_2 & \omega x_2 & -\omega x_3 & -\omega x_4 & -\omega^2 x_1 & x_1 & -2 x_5
\end{array}
\right )
},\label{g1}
\end{equation}
%----g(2)------------------------------------------------------------------
\begin{equation}
X^{(2)}=
{\footnotesize
\left (
\begin{array}{*{9}{@{}C{\mycolwd}@{}}}
0 & x_1 & -\omega x_1 & -\omega^2 x_4 & -\omega^2 x_3 & \omega_2 x_2 & \omega x_2 & 0 \\
-x_1 & 0 & x_2 & \omega x_3 & \omega x_4 & \omega^2 x_1 & 0 & \omega x_2 \\
\omega x_1 & -x_2 & 0 & x_3 & x_4 & 0 & \omega^2 x_1 & -\omega^2 x_2  \\
\omega^2 x_4 & -\omega x_3 & -x_3 & 0 & 0 & x_4 & -\omega x_4 & -\omega^2 x_3 \\
\omega^2 x_3 & -\omega x_4 & -x_4 & 0 & 0 & x_3 & -\omega x_3 & -\omega^2 x_4 \\
-\omega^2 x_2 & -\omega^2 x_1 & 0 & -x_4 & -x_3 & 0 & x_2 & \omega x_1 \\
-\omega x_2 & 0 & -\omega^2 x_1 & \omega x_4 & \omega x_3 & -x_2 & 0 & x_1 \\
0 & -\omega x_2 & \omega^2 x_2 & \omega^2 x_3 & \omega^2 x_4 & -\omega x_1 & -x_1 & 0
\end{array}
\right )
},\label{g2}
\end{equation}
%----g(3)------------------------------------------------------------------
\begin{equation}
X^{(3)}=
{\footnotesize
\left (
\begin{array}{*{9}{@{}C{\mycolwd}@{}}}
2x_5 & x_1 & x_1 & x_4 & x_3 & x_2 & x_2 & 0 \\
x_1 & -2x_5 & x_2 & -x_3 & -x_4 & -x_1 & 0 & x_2 \\
x_1 & x_2 & -2x_5 & x_3 & x_4 & 0 & -x_1 & -x_2  \\
x_4 & -x_3 & x_3 & 6x_6 & 0 & x_4 & x_4 & x_3 \\
x_3 & -x_4 & x_4 & 0 & -6x_6 & x_3 & x_3 & x_4 \\
x_2 & -x_1 & 0 & x_4 & x_3 & 2x_5 & x_2 & -x_1 \\
x_2 & 0 & -x_1 & x_4 & x_3 & x_2 & 2 x_5 & x_1 \\
0 & x_2 & -x_2 & x_3 & x_4 & -x_1 & x_1 & -2 x_5
\end{array}
\right )
},\label{g3}
\end{equation}
%----g(4)------------------------------------------------------------------
\begin{equation}
X^{(4)}=
{\footnotesize
\left (
\begin{array}{*{9}{@{}C{\mycolwd}@{}}}
0 & x_1 & \omega^2 x_1 & \omega x_4 & \omega x_3 & -\omega x_2 & -\omega^2 x_2 & 0 \\
-x_1 & 0 & x_2 & -\omega^2 x_3 & -\omega^2 x_4 & -\omega x_1 & 0 & -\omega^2 x_2 \\
-\omega^2 x_1 & -x_2 & 0 & x_3 & x_4 & 0 & -\omega x_1 & \omega x_2  \\
-\omega x_4 & \omega^2 x_3 & -x_3 & 0 & 0 & x_4 & \omega^2 x_4 & \omega x_3 \\
-\omega x_3 & \omega^2 x_4 & -x_4 & 0 & 0 & x_3 & \omega^2 x_3 & \omega x_4 \\
\omega x_2 &  \omega x_1  & 0 & -x_4 & -x_3 & 0 & x_2 & -\omega^2 x_1 \\
\omega^2 x_2 & 0 & \omega x_1 & -\omega^2 x_4 & -\omega^2 x_3 & -x_2 & 0 & x_1 \\
0 & \omega2 x_2 & -\omega x_2 & -\omega x_3 & -\omega x_4 & \omega^2 x_1 & -x_1 & 0
\end{array}
\right )
},\label{g4}
\end{equation}
%----g(5)------------------------------------------------------------------
\begin{equation}
X^{(5)}=
{\footnotesize
\left (
\begin{array}{*{9}{@{}C{\mycolwd}@{}}}
2 x_5 & x_1 & -\omega x_1 & \omega^2 x_4 & \omega^2 x_3 & \omega^2 x_2 & -\omega x_2 & 0 \\
x_1 & -2 \omega^2 x_5 & x_2 & \omega x_3 & \omega x_4 & -\omega^2 x_1 & 0 & -\omega x_2 \\
-\omega x_1 & x_2 & 2 \omega x_5 & x_3 & x_4 & 0 & -\omega^2 x_1 & -\omega^2 x_2  \\
\omega^2 x_4 & \omega x_3 & x_3 & 0 & 0 & x_4 & -\omega x_4 & \omega^2 x_3 \\
\omega^2 x_3 & \omega x_4 & x_4 & 0 & 0 & x_3 & -\omega x_3 & \omega^2 x_4 \\
\omega^2 x_2 & -\omega^2 x_1 & 0 & x_4 & x_3 & -2\omega x_5 & x_2 & \omega x_1 \\
-\omega x_2 & 0 & -\omega^2 x_1 & -\omega x_4 & -\omega x_3 & x_2 & 2 \omega^2 x_5 & x_1 \\
0 & -\omega x_2 & -\omega^2 x_2 & \omega^2 x_3 & \omega^2 x_4 & \omega x_1 & x_1 & -2 x_5
\end{array}
\right )
}.\label{g5}
\end{equation}
Here $\omega= e^{\frac{2 \pi i}{6}} = \frac{1}{2} + i \frac{1}{2} \sqrt{3}$.

%-----The Inverse of AdJ----------------------------------------------
\section{The inverse of $\mbox{ad}_J$}

Since $J$ is an element of the Cartan subalgebra $\mathfrak{h}$, the eigenvalues of $\mbox{ ad}_J$ are given by:
\begin{equation}
\big[  J ,  E_{\alpha} \big] = \alpha(J)   E_{\alpha},
\end{equation}
where the root $\alpha \in \Delta$. The root system $\Delta$ of $D_4$ splits into positive and negative roots $\Delta=\Delta^+\cup \Delta^-$,
so that if $\alpha \in \Delta^+ $ then $-\alpha \in \Delta^- $.
In particular
\[ \Delta^+ \equiv \{  e_i \pm e_j, \quad i,j=1..4,\  i<j \}. \]
The eigenvalues of $\ad_J$ are given by:
\[ \alpha(J) \equiv (\alpha, \vec{J}), \qquad \alpha \in \Delta ,\]
where $\vec{J}$ is the vector dual to the Cartan element $J$ (see \eqref{J}), so $\vec{J} = e_1+ \omega e_2 + \omega^5 e_3$. As a result, the
characteristic polynomial of $\ad_J$ is:
\begin{equation}\label{eq:cpJ}\begin{split}
P(\lambda) \equiv \prod_{\alpha\in\Delta}^{} (\lambda -\alpha(J)) =  \prod_{\alpha\in\Delta^+}^{} (\lambda^2 -(\alpha(J))^2).
\end{split}\end{equation}
Using the explicit form of $\vec{J}$ and $\Delta^+$, and the fact that $\omega=\exp(2\pi i/6)$ after some simplifications we get:
\begin{equation}\label{eq:cpJ2}\begin{split}
 P(\lambda) = (\lambda^6 -1)^3 (\lambda^6 +27).
\end{split}\end{equation}

Next we note that $\ad_J$ has an inverse only on its image, i.e. $\ad_J^{-1}$ is defined on $D_4\backslash \mathfrak{h}_{D_4}$. On this subspace of
$D_4$ $\ad_J$ satisfies the minimal characteristic polynomial:
\begin{equation}
\ad_J^{12} + 26\ad_J^6 - 27 = 0.
\end{equation}
Thus we find:
\begin{equation}
\mbox{ad}_J^{-1}=\frac{1}{27} \left(26 \mbox{ ad}_J^{5}+ \mbox{ad}_J^{11} \right).
\end{equation}

\end{document}